\documentclass[pra,twocolumn,superscriptaddress,10pt]{revtex4-1}
\usepackage{amsmath,amssymb,amsthm,mathtools}
\usepackage{bm}
\usepackage[colorlinks=true,citecolor=blue,urlcolor=blue]{hyperref}

\newtheorem{theorem}{Theorem}

\newtheorem{lemma}{Lemma}

\newcommand{\EF}{E_{\mathrm F}}
\newcommand{\C}{\mathcal C}
\newcommand{\supp}{\operatorname{supp}}
\newcommand{\Rad}{\operatorname{Rad}}
\newcommand{\tr}{\operatorname{Tr}}
\newcommand{\T}{\mathsf T}

\begin{document}

\title{Supporting functionals and singular boundary geometry of two-qubit entanglement of formation}
\author{Wei Song}\email{wsong1@mail.ustc.edu.cn}
\affiliation{School of Physics and Materials Engineering, Hefei Normal University,
Hefei 230601, China}

\author{Xiao-Lan Zong}\email{zxl@hfnu.edu.cn}
\affiliation{School of Physics and Materials Engineering, Hefei Normal University,
Hefei 230601, China}
\date{\today}

\begin{abstract}
We give a criterion for the existence of a global supporting affine
functional for the entanglement of formation at an entangled
two-qubit state. It exists if and only if the spin-flip bilinear form
is nondegenerate on the support of the state, or equivalently, if the
last concurrence root associated with the support is nonzero. In this case the entanglement of formation satisfies a two-sided
Lipschitz bound relative to the reference state.
In the degenerate case the state has rank two or three, and its
projective support is tangent to the product-state quadric. By suitably choosing a perturbation outside the original support space, the decrease in the entanglement of formation is, to leading order, proportional to the square root of the mixing parameter. We give an explicit expression for the leading coefficient. Such a square-root decrease violates the one-sided Lipschitz bound, and therefore no globally supported affine functional exists at this state.
\end{abstract}

\maketitle

\section{Introduction}

Entanglement of formation is defined as the minimum average pure-state
entanglement over all decompositions of a bipartite mixed state
\cite{Bennett1996}. The optimization over pure-state decompositions is
generally difficult. Two-qubit systems are an exception: following the
introduction of concurrence by Hill and Wootters \cite{Hill1997},
Wootters derived an analytical formula for the entanglement of formation
of an arbitrary two-qubit state \cite{Wootters1998}.

For rank-two states, the zeros of the entanglement measure in the
support play an important role in the convex-roof construction.
Osterloh, Siewert, and Uhlmann introduced the zero polytope for
polynomial entanglement measures and used characteristic curves to
obtain lower bounds for the corresponding convex roofs
\cite{OsterlohSiewertUhlmann2008}. Regula and Adesso proved that, when
the support contains a unique pure state at which the measure vanishes,
the average of a degree-two homogeneous entanglement measure is
independent of the pure-state decomposition \cite{RegulaAdesso2016PRL}.
Using the positions of the zeros on the Bloch sphere, they also
obtained a geometric description of concurrence and analytical
convex-roof results for supports containing one or two distinct zeros
\cite{RegulaAdesso2016PRA}. Osterloh used intersections with the zero
polytope to determine where entanglement vanishes in rank-two mixtures
\cite{Osterloh2016}. For two qubits, a pure state has zero concurrence
if and only if it is a product state. The one-root condition thus means
that the support contains a unique product ray, corresponding to a
projective line tangent to the product-state quadric.

We consider the problem of existence of a global supporting affine
functional for the entanglement of formation at a given state. Such a
functional is an affine lower bound for the entanglement of formation
which is attained at the reference state. For full-rank states its
existence follows from standard results of convex analysis. For
degenerate states this is no longer guaranteed. Holevo and Shirokov
recently showed that the existence of a global supporting affine
functional is equivalent to a one-sided Lipschitz bound for the
entanglement of formation and gave a rank-two two-qubit state for which
this bound fails \cite{HolevoShirokov2026}. It remains to determine
which boundary states admit a global supporting affine functional.

Previous studies on zeros and convex roofs have mostly been limited to a fixed support space, whereas a globally supported affine functional must hold for all quantum states. Even if the convex-roof form of entanglement of formation is simple within the original support space, perturbations away from that space may still change its local behavior. A single-root state is an example: all its pure-state decompositions yield the same average concurrence, but this property does not guarantee that a globally supported affine functional exists for entanglement of formation.

We obtain a necessary and sufficient condition for the existence of a
global supporting affine functional for the entanglement of formation
at an entangled two-qubit state. Such a functional exists exactly when
the spin-flip bilinear form is nondegenerate on the support of the
state. In rank two, the degenerate case is the one-root class. In rank
three, degeneracy occurs exactly when the kernel of the state is
spanned by a product vector. For every degenerate state we construct a
perturbation outside the support for which the entanglement of formation
decreases as the square root of the mixing parameter. The leading coefficient
is obtained explicitly and is invariant under local unitary transformations.
In the nondegenerate case we prove a Lipschitz bound for all sufficiently
small perturbations of the state. The two arguments are based on the
same block unitary-congruence lemma.

\section{Concurrence and the support form}

Let \(\mathcal H=\mathbb C^2\otimes\mathbb C^2\) and fix the computational product basis in \(\mathcal H\). We use the
spin-flip antiunitary
\begin{equation}
 \Theta=JK,\qquad J=\sigma_y\otimes\sigma_y,
 \label{eq:theta}
\end{equation}
where \(K\) is complex conjugation in this basis. Since \(J\) is real,
symmetric, and unitary, one has \(\Theta^2=I\).

Define the symmetric complex bilinear form
\begin{equation}
 \beta(x,y)=x^\T Jy.
 \label{eq:beta}
\end{equation}

Moreover, \(\langle x|\Theta y\rangle=\overline{\beta(x,y)}\). If \(x=\sum_{i,j=0}^1 x_{ij}|ij\rangle\)
and \(X=(x_{ij})\) is the corresponding \(2\times2\) coefficient matrix, then
\begin{equation}
 \beta(x,x)=-2\det X.
 \label{eq:productquadric}
\end{equation}

Thus \(\beta(x,x)=0\) exactly for product vectors \(x\). The product vectors form the Segre quadric
\(Q=\{[x]\in\mathbb P(\mathcal H):\beta(x,x)=0\}\).
For a state \(\rho\), write \(\widetilde\rho=J\overline{\rho}J\). Let \(\lambda_1\geq\lambda_2\geq\lambda_3\geq\lambda_4\geq0\) be the square roots of the eigenvalues of \(\rho\widetilde\rho\),
counted with multiplicity. Wootters' formula \cite{Wootters1998}
takes the form
\begin{align}
 \C(\rho)
 &=\max\{0,\lambda_1-\lambda_2-\lambda_3-\lambda_4\},
 \label{eq:concurrence}\\
 \EF(\rho)
 &=\mathcal E(\C(\rho)),
 \qquad
 \mathcal E(c)
 =h_2\!\left(\frac{1+\sqrt{1-c^2}}{2}\right),
 \label{eq:eof}
\end{align}
where \(h_2(p)=-p\log_2p-(1-p)\log_2(1-p)\). The function \(\mathcal E\) is smooth and strictly increasing on
\((0,1)\). In particular,
\begin{equation}
 \mathcal E'(c)>0,\qquad 0<c<1,
 \qquad
 \mathcal E'(1^-)=\frac{1}{\ln2}.
 \label{eq:Ederivative}
\end{equation}

Take a four-column factor \(X\) of \(\rho\) satisfying \(\rho=XX^\dagger\), with zero columns added when necessary. Consider the complex symmetric matrix
\begin{equation}
 T(X)=X^\T JX.
 \label{eq:takagimatrix}
\end{equation}

Its singular values are precisely \(\lambda_1,\ldots,\lambda_4\). Indeed, \(T(X)^\dagger T(X)=X^\dagger\widetilde\rho X\), and \(X^\dagger\widetilde\rho X\) and \(XX^\dagger\widetilde\rho=\rho\widetilde\rho\) have the same nonzero eigenvalues. If \(X'\) is another four-column factor of \(\rho\), then \(X'=XU\) for some unitary \(U\), and hence \(T(X')=U^\T T(X)U\). Thus the singular values of \(T(X)\) depend only on \(\rho\).

Let \(S=\supp\rho\) and \(r=\dim S\). The restriction of \(\beta\) to \(S\) has radical
\begin{equation}
 \Rad(\beta|_S)
 =
 \{z\in S:\beta(z,y)=0\ \text{for all }y\in S\}.
 \label{eq:rad}
\end{equation}

Choose a full-column-rank operator \(X_S:\mathbb C^r\longrightarrow S\) such that \(\rho=X_SX_S^\dagger\).
In the basis of \(S\) formed by the columns of \(X_S\), the form
\(\beta|_S\) is represented by \(X_S^\T JX_S\). Hence its singular values are \(\lambda_1(\rho),\ldots,\lambda_r(\rho)\)
while the remaining concurrence roots vanish. Therefore,
\begin{equation}
 \Rad(\beta|_S)=\{0\}
 \quad\Longleftrightarrow\quad
 \lambda_r(\rho)>0.
 \label{eq:rankroots}
\end{equation}

Here \(\lambda_r(\rho)\) denotes the \(r\)th concurrence root and not an
eigenvalue of \(\rho\). A functional \(\sigma\longmapsto\tr(\Lambda\sigma)\),
with \(\Lambda\) Hermitian, is called supporting for \(\EF\) at \(\rho\) if
\begin{equation}
 \tr(\Lambda\rho)=\EF(\rho)
 \quad\text{and}\quad
 \tr(\Lambda\sigma)\leq\EF(\sigma)
 \label{eq:supportdef}
\end{equation}
for all two-qubit states \(\sigma\). By the finite-dimensional form of
the Holevo--Shirokov criterion
\cite[Theorem 1]{HolevoShirokov2026}, such a functional exists if and
only if there is a finite constant \(L_\rho\) such that
\begin{equation}
 \EF(\rho)-\EF(\sigma)
 \leq
 L_\rho\|\rho-\sigma\|_1
 \qquad
 \text{for every state }\sigma.
 \label{eq:lowerlip}
\end{equation}

It is enough to establish \eqref{eq:lowerlip} in a neighborhood of
\(\rho\). Indeed, if the inequality holds whenever
\(\|\rho-\sigma\|_1<\varepsilon\), then \(0\leq\EF\leq1\) shows that,
after replacing \(L_\rho\) by \(\max\{L_\rho,\varepsilon^{-1}\}\), it holds for all states.

Thus, only a one-sided Lipschitz bound for the decrease of \(\EF\) is
required. No analogous bound for its increase is needed.

\section{Main result}

\begin{theorem}[Complete classification on the entangled set]
\label{thm:main}
Let \(\rho\) be an entangled two-qubit state with support
\(S=\supp\rho\) and rank \(r=\dim S\). A global supporting affine
functional for \(\EF\) exists at \(\rho\) if and only if
\(\beta|_S\) is nondegenerate, or equivalently,
\(\lambda_r(\rho)>0\). In this case, there are constants
\(L_\rho>0\) and \(\varepsilon_\rho>0\) such that
\begin{equation}
 |\EF(\sigma)-\EF(\rho)|
 \leq L_\rho\|\sigma-\rho\|_1
 \label{eq:pointwiselip}
\end{equation}
for every two-qubit state \(\sigma\) satisfying
\(\|\sigma-\rho\|_1<\varepsilon_\rho\).

If \(\beta|_S\) is degenerate, then \(r=2\) or \(r=3\), and its
radical is one-dimensional. Let \(z\) be a unit vector spanning
\(\Rad(\beta|_S)\), set \(w=\Theta z\), and denote by
\(\rho_S^{-1}\) the inverse of \(\rho\) on \(S\).
Then \(w\in S^\perp\). Define
\begin{equation}
 \kappa_\rho=
 \bigl\langle z\big|\rho_S^{-1}\big|z\bigr\rangle^{-1/2}>0.
 \label{eq:kappa}
\end{equation}

Along the path \(\rho_t=(1-t)\rho+t|w\rangle\langle w|\),
we have, as \(t\to0^+\),
\begin{align}
 \C(\rho_t)
 &=\C(\rho)-2\kappa_\rho\sqrt t+O(t),
 \label{eq:cuspC}\\
 \EF(\rho)-\EF(\rho_t)
 &=2\mathcal E'(\C(\rho))\kappa_\rho\sqrt t+O(t).
 \label{eq:cuspE}
\end{align}

The coefficient of the square-root term in Eq.~\eqref{eq:cuspE}
is strictly positive and unchanged by local unitary transformations.

The restricted form is degenerate for a rank-two state if and only if the projective line \(\mathbb P(S)\) is tangent to \(Q\). At rank three, degeneracy is equivalent to \(\ker\rho\) being spanned by a product vector, or geometrically, to the projective plane \(\mathbb P(S)\) being tangent to \(Q\). Entangled states of rank one or four always have a nondegenerate restricted form.
\end{theorem}

For separable states, since \(E_F \geq 0\) and \(E_F = 0\) at such a state, the zero functional is a supporting functional for \(E_F\) at that state, and no additional conditions are required. Below we give the lemmas that will be used in the proof.

\section{A block lemma for the concurrence roots}

We will use the following lemma to estimate the concurrence roots
near zero. It applies without a differentiable dependence on the
perturbation parameter or a smooth choice of Takagi vectors.

\begin{lemma}[Unitary-congruence block reduction]
\label{lem:block}
Let \(M_\varepsilon\) be a complex symmetric matrix with a fixed
block decomposition
\begin{equation}
 M_\varepsilon=
 \begin{pmatrix}
 A+O(\varepsilon^2)&\varepsilon B_\varepsilon+O(\varepsilon^2)\\
 \varepsilon B_\varepsilon^\T+O(\varepsilon^2)&
 \varepsilon C+O(\varepsilon^2)
 \end{pmatrix},
 \label{eq:blockhyp}
\end{equation}
where \(A=A^\T\in\mathbb C^{m\times m}\) is nonsingular,
\(C=C^\T\in\mathbb C^{n\times n}\), and both \(A\) and \(C\)
are fixed. Assume that \(B_\varepsilon\) is uniformly bounded.
All \(O\)-estimates are in the operator norm as
\(\varepsilon\to0^+\).

There exists a unitary matrix \(W_\varepsilon=I+O(\varepsilon)\)
such that
\begin{equation}
 W_\varepsilon^\T M_\varepsilon W_\varepsilon=
 \begin{pmatrix}
 A+O(\varepsilon^2)&O(\varepsilon^2)\\
 O(\varepsilon^2)&\varepsilon C+O(\varepsilon^2)
 \end{pmatrix}.
 \label{eq:blocknormal}
\end{equation}

If the singular values are ordered nonincreasingly, then
\begin{align}
 s_j(M_\varepsilon)
 &=s_j(A)+O(\varepsilon^2),
 &&1\leq j\leq m,\\
 s_{m+k}(M_\varepsilon)
 &=\varepsilon s_k(C)+O(\varepsilon^2),
 &&1\leq k\leq n.
\end{align}

For fixed \(A\) and \(C\), the same constants may be used when
\(B_\varepsilon\) and the remainders in
Eq.~\eqref{eq:blockhyp} satisfy common bounds.
\end{lemma}

\begin{proof}
We first prove Eq.~\eqref{eq:blocknormal}.
Put \(Y_\varepsilon=-A^{-1}B_\varepsilon\) and define
\begin{equation}
 H_\varepsilon=
 \begin{pmatrix}
 0&Y_\varepsilon\\
 -Y_\varepsilon^\dagger&0
 \end{pmatrix},
 \qquad
 W_\varepsilon=\exp(\varepsilon H_\varepsilon).
 \label{eq:generator}
\end{equation}

The matrix \(H_\varepsilon\) is anti-Hermitian and uniformly
bounded. Hence \(W_\varepsilon\) is unitary and
\(W_\varepsilon=I+O(\varepsilon)\).

Let \(M_0=A\oplus0\). Since
\(M_\varepsilon-M_0=O(\varepsilon)\), expansion of the
matrix exponentials gives
\[
 W_\varepsilon^\T M_\varepsilon W_\varepsilon
 =
 M_\varepsilon+
 \varepsilon\bigl(
 H_\varepsilon^\T M_0+M_0H_\varepsilon
 \bigr)+O(\varepsilon^2).
\]
Here
\[
 H_\varepsilon^\T M_0+M_0H_\varepsilon
 =
 \begin{pmatrix}
 0&AY_\varepsilon\\
 Y_\varepsilon^\T A&0
 \end{pmatrix}
 =
 \begin{pmatrix}
 0&-B_\varepsilon\\
 -B_\varepsilon^\T&0
 \end{pmatrix},
\]
where the second equality follows from the definition of
\(Y_\varepsilon\) and the symmetry of \(A\).
The \(O(\varepsilon)\) terms in the off-diagonal blocks cancel.
There is no correction of this order in the diagonal blocks.
This proves Eq.~\eqref{eq:blocknormal}.

Unitary congruence preserves singular values, and
\(\bigl\|W_\varepsilon^\T M_\varepsilon W_\varepsilon
-(A\oplus\varepsilon C)\bigr\|=O(\varepsilon^2)\).
The inequality
\(|s_j(P)-s_j(Q)|\leq\|P-Q\|\)
therefore implies
\(s_j(M_\varepsilon) =s_j(A\oplus\varepsilon C)+O(\varepsilon^2).\)
Since \(s_{\min}(A)>0\), we have
\(\varepsilon\|C\|<s_{\min}(A)\) for sufficiently small
\(\varepsilon>0\).
The first \(m\) singular values of \(A\oplus\varepsilon C\)
are those of \(A\), and the remaining \(n\) are
\(\varepsilon\) times those of \(C\). This remains true in the
presence of degeneracies within either block.

These bounds depend only on \(A\), \(C\), the bound on
\(B_\varepsilon\), and the remainders in Eq.~\eqref{eq:blockhyp}.
This also proves the uniformity assertion.
\end{proof}

\section{Nondegenerate supports: uniform stability}
\label{sec:regular}

Let \(\rho\) be an entangled two-qubit state such that
\(\beta|_S\) is nondegenerate, where \(S=\supp\rho\).
We show that Eq.~\eqref{eq:pointwiselip} holds in a sufficiently
small neighborhood of \(\rho\).

Choose orthonormal bases for \(S\) and \(S^\perp\), and arrange the basis vectors as columns of the unitary matrix \(U=(U_S,U_{S^\perp})\). Then
\[
 U^\dagger\rho U=
 \begin{pmatrix}
 A_0&0\\
 0&0
 \end{pmatrix},
 \qquad A_0>0.
\]
For a state \(\sigma\), put \(\delta=\|\sigma-\rho\|_1\) and write
\[
 U^\dagger\sigma U=
 \begin{pmatrix}
 A&B\\
 B^\dagger&D
 \end{pmatrix}.
\]
The block compressions of \(\sigma-\rho\) satisfy
\(\|A-A_0\|\leq\delta,\qquad \|B\|\leq\delta,\qquad \|D\|\leq\delta.\)
Let \(a_0=\lambda_{\min}(A_0)>0\).
If \(\delta<a_0/2\), then \(A\geq(a_0/2)I\).
Thus \(A^{-1}\) and \(A^{-1/2}\) are uniformly bounded. All estimates
below are uniform for \(\sigma\) in a sufficiently
small neighborhood of \(\rho\).

Since \(\sigma\geq0\), the Schur complement
\(R=D-B^\dagger A^{-1}B\)
satisfies \(0\leq R\leq D\). Hence
\begin{equation}
 X_\sigma=
 \begin{pmatrix}
 A^{1/2}&0\\
 B^\dagger A^{-1/2}&R^{1/2}
 \end{pmatrix},
 \qquad
 X_\sigma X_\sigma^\dagger=U^\dagger\sigma U.
 \label{eq:schur}
\end{equation}

The common positive lower bound for \(A\) and \(A_0\) gives
\(A^{1/2}-A_0^{1/2}=O(\delta), \qquad B^\dagger A^{-1/2}=O(\delta).\)
Here we use the operator-Lipschitz continuity of the square-root map
on positive definite matrices bounded away from zero.
Also, \(0\leq R\leq D\) gives \(R^{1/2}=O(\sqrt{\delta})\).
Thus the first block column of \(X_\sigma\) differs from that of \(X_\rho=\operatorname{diag}(A_0^{1/2},0)\) by \(O(\delta)\), while the second block column is \(O(\sqrt{\delta})\).

By Eq.~\eqref{eq:takagimatrix}, the concurrence roots of
\(\sigma\) are the singular values of \(T_\sigma:=T(UX_\sigma)=X_\sigma^\T U^\T JU X_\sigma\). Hence
\begin{equation}
\begin{aligned}
 T_\sigma&=
 \begin{pmatrix}
 T_0+O(\delta)&O(\sqrt{\delta})\\
 O(\sqrt{\delta})&O(\delta)
 \end{pmatrix},\\
 T_0&=(A_0^{1/2})^\T J_{SS}A_0^{1/2},
 \qquad J_{SS}=U_S^\T JU_S.
\end{aligned}
 \label{eq:regblock}
\end{equation}

The matrix \(J_{SS}\) represents \(\beta|_S\) and is nonsingular by
assumption. Since \(A_0>0\), the matrix \(T_0\) is also nonsingular.

For \(r<4\), apply Lemma~\ref{lem:block} to
Eq.~\eqref{eq:regblock} with
\(\varepsilon=\sqrt{\delta}\) and \(C=0\).
The positive concurrence roots then change by \(O(\delta)\), while
the remaining roots are \(O(\delta)\). Hence there is a constant
\(K_\rho<+\infty\) such that
\begin{equation}
 |\lambda_j(\sigma)-\lambda_j(\rho)|
 \leq K_\rho\delta,
 \qquad j=1,\ldots,4.
 \label{eq:rootlip}
\end{equation}

For \(r=4\), Eq.~\eqref{eq:regblock} reduces to
\(T_\sigma=T_0+O(\delta)\), and Eq.~\eqref{eq:rootlip}
follows directly from the singular-value perturbation bound.

We now estimate the change in \(\EF\). The inequality
\(|\max\{0,x\}-\max\{0,y\}|\leq|x-y|\)
and Eqs.~\eqref{eq:concurrence} and \eqref{eq:rootlip} give
\[
 |\C(\sigma)-\C(\rho)|
 \leq\sum_{j=1}^4|\lambda_j(\sigma)-\lambda_j(\rho)|
 \leq4K_\rho\delta.
\]

Since \(\rho\) is entangled, \(\C(\rho)>0\).
The derivative of \(\mathcal E\) is bounded in a sufficiently small
neighborhood of \(\C(\rho)\) in \([0,1]\). This is also true at
\(\C(\rho)=1\), since \(\mathcal E'(1^-)=1/\ln2\). Therefore, for some finite constant \(M_\rho\),
\[
 \begin{aligned}
 |\EF(\sigma)-\EF(\rho)|
 &\leq M_\rho|\C(\sigma)-\C(\rho)|\\
 &\leq4M_\rho K_\rho\|\sigma-\rho\|_1
 \end{aligned}
\]
for all \(\sigma\) sufficiently close to \(\rho\).
This proves Eq.~\eqref{eq:pointwiselip}.

The bound \(0\leq\EF\leq1\) allows the constant to be enlarged so that
the one-sided estimate holds for all states. Therefore,
Eq.~\eqref{eq:lowerlip} holds. By the Holevo--Shirokov criterion,
a global supporting affine functional exists at \(\rho\).

\section{Tangent supports and the square-root cusp}

\subsection{Geometry of the degenerate cases}

We describe the supports \(S\) for which \(\beta|_S\) is degenerate,
assuming throughout that \(\rho\) is entangled. Recall that the
tangent plane to \(Q\) at a product ray \([z]\) is
\(\mathbb P(T_zQ)\), where
\begin{equation}
 T_zQ=\{y\in\mathcal H:\beta(z,y)=0\}.
 \label{eq:tangent}
\end{equation}
This follows by differentiating \(\beta(x,x)=0\) at \(z\).

\emph{Rank two.}
Assume that \(\beta|_S\) is degenerate. Its rank cannot be zero.
Otherwise every nonzero vector in \(S\)
would be a product vector, and every state supported on \(S\)
would be separable. Therefore, \(\beta|_S\) has rank one.

Take a nonzero \(z\) in the radical and complete it to a basis
\(\{z,u\}\) of \(S\). Then
\(\beta(z,z)=\beta(z,u)=0\), while \(\beta(u,u)\neq0\).
Consequently,
\(\beta(az+bu,az+bu)=b^2\beta(u,u).\)
It follows that \([z]\) is the only product ray in
\(\mathbb P(S)\), with intersection multiplicity two.
Thus \(\mathbb P(S)\) is tangent to \(Q\) at \([z]\)
and is not contained in \(Q\). Conversely, if \(\mathbb P(S)\) is tangent to \(Q\) at \([z]\),
then \(\beta(z,y)=0\) for every \(y\in S\), so
\(\beta|_S\) is degenerate.

\emph{Rank three.}
Take a unit vector \(k\) spanning \(\ker\rho\); then
\(S=k^\perp\). From
\(\beta(z,y)=\langle\Theta z|y\rangle\),
we obtain
\(z\in\Rad(\beta|_S)\) if and only if
\(\Theta z\in\operatorname{span}\{k\}\).
Since \(\Theta^2=I\), we obtain
\(\Rad(\beta|_S)=S\cap\operatorname{span}\{\Theta k\}.\)
Hence \(\beta|_S\) is degenerate if and only if \(\Theta k\in S\),
or equivalently,
\begin{equation}
 \langle k|\Theta k\rangle=0
 \quad\Longleftrightarrow\quad
 \beta(k,k)=0.
 \label{eq:kernelproduct}
\end{equation}
The latter condition means that \(k\) is a product vector.
In this case, the radical is one-dimensional and is spanned
by the product vector \(z=\Theta k\). Moreover,
\(\beta(z,y)=\langle k|y\rangle\), which implies
\(S=T_zQ\). Thus \(\mathbb P(S)\) is the tangent plane
to \(Q\) at \([z]\).

\emph{Ranks one and four.}
If \(S=\operatorname{span}\{\psi\}\) and \(\rho\) is entangled,
then \(\beta(\psi,\psi)\neq0\), so \(\beta|_S\) is nondegenerate.
For rank four, \(S=\mathcal H\), and nondegeneracy follows
from the invertibility of \(J\).

\subsection{Transverse perturbation}

Let \(\rho\) be an entangled state for which \(\beta|_S\) is
degenerate. By the preceding classification, \(r=2\) or \(r=3\),
and \(\Rad(\beta|_S)\) is one-dimensional. By the
Holevo--Shirokov criterion, to prove the absence of a global
supporting affine functional it suffices to find a path
\(\rho_t\to\rho\) such that
\[
 \frac{\EF(\rho)-\EF(\rho_t)}{\|\rho_t-\rho\|_1}
 \longrightarrow+\infty.
\]

Choose a unit vector \(z\) spanning \(\Rad(\beta|_S)\),
and put \(w=\Theta z\). Since
\begin{equation}
 \langle\Theta z|y\rangle=\beta(z,y)=0
 \quad(y\in S),
 \label{eq:orthogonal}
\end{equation}
\(w\) is a unit vector in \(S^\perp\).
Both \(z\) and \(w\) are product vectors, so \(\beta(w,w)=0\).
Consider the states \(\rho_t=(1-t)\rho+t|w\rangle\langle w|\), \(0\leq t\leq1\). The orthogonality of the supports of \(\rho\) and
\(|w\rangle\langle w|\) gives
\begin{equation}
 \|\rho_t-\rho\|_1=2t.
 \label{eq:tracedistance}
\end{equation}

Write \(\rho_S=X_0X_0^\dagger\), where
\(X_0:\mathbb C^r\to S\) is invertible as a map onto \(S\),
and put \(T_0=X_0^\T JX_0\). Since \(X_0\) maps \(\ker T_0\) onto \(\Rad(\beta|_S)\),
the kernel of \(T_0\) is spanned by the unit vector
\begin{equation}
 a=\frac{X_0^{-1}z}{\|X_0^{-1}z\|},
 \qquad
 \|X_0^{-1}z\|^2
 =\langle z|\rho_S^{-1}|z\rangle.
 \label{eq:zeromode}
\end{equation}

Thus \(X_0a=\kappa_\rho z\), with \(\kappa_\rho\) defined
in Eq.~\eqref{eq:kappa}.
Complete \(a\) to a unitary matrix \(V=(V_+,a)\). Then
\begin{equation}
 V^\T T_0V=
 \begin{pmatrix}
 A&0\\
 0&0
 \end{pmatrix},
 \qquad
 A=A^\T\in\mathbb C^{(r-1)\times(r-1)},
 \label{eq:singularT0}
\end{equation}
where \(A\) is nonsingular, since \(\ker T_0\) is one-dimensional.

The choice \(w=\Theta z\) gives
\begin{equation}
 a^\T X_0^\T Jw
 =(X_0a)^\T J\Theta z
 =\kappa_\rho z^\T\overline z
 =\kappa_\rho.
 \label{eq:cross}
\end{equation}

Take \(X_t=(\sqrt{1-t}\,X_0V,\sqrt t\,w)\),
with a zero column added when \(r=2\). Then
\begin{equation}
 T(X_t)=
 \begin{pmatrix}
 (1-t)A&0&\varepsilon b\\
 0&0&\varepsilon\kappa_\rho\\
 \varepsilon b^\T&\varepsilon\kappa_\rho&0
 \end{pmatrix}
 \oplus 0_{3-r},
 \label{eq:singularblock}
\end{equation}
where \(\varepsilon=\sqrt{t(1-t)}\) and \(b=V_+^\T X_0^\T Jw\).
The last diagonal entry of the displayed matrix is \(t\beta(w,w)=0\).
The additional zero block is omitted when \(r=3\).

To apply Lemma~\ref{lem:block}, group the last two coordinates
of the displayed matrix together. The lower-right block is then
\(\varepsilon C_{\mathrm{eff}}\), where
\begin{equation}
 C_{\mathrm{eff}}=
 \begin{pmatrix}
 0&\kappa_\rho\\
 \kappa_\rho&0
 \end{pmatrix},
 \label{eq:effective}
\end{equation}
and the upper-right block is \(\varepsilon(0,b)\).
Moreover, \(t=O(\varepsilon^2)\) as \(t\to0^+\), so \((1-t)A=A+O(\varepsilon^2)\).
Hence Lemma~\ref{lem:block} applies.

Let \(\mu_1\geq\cdots\geq\mu_{r-1}>0\) be the singular values of \(A\), which are the positive concurrence
roots of \(\rho\). The two singular values of
\(C_{\mathrm{eff}}\) are both \(\kappa_\rho\). Therefore,
Lemma~\ref{lem:block}, together with \(\varepsilon=\sqrt t+O(t^{3/2})\), gives
\begin{align}
 \lambda_j(\rho_t)
 &=\mu_j+O(t),
 &&1\leq j\leq r-1,\\
 \lambda_j(\rho_t)
 &=\kappa_\rho\sqrt t+O(t),
 &&j=r,r+1.
\end{align}
For \(r=2\),
\(\lambda_4(\rho_t)=0.\)
For sufficiently small \(t\), the first \(r-1\) roots are
greater than \(\mu_{r-1}/2\), while the two new roots are
less than this number. Hence Eq.~\eqref{eq:concurrence} gives
\[
 \lambda_1(\rho_t)-\sum_{j=2}^{4}\lambda_j(\rho_t)
 =\C(\rho)-2\kappa_\rho\sqrt t+O(t).
\]
Since \(\C(\rho)>0\), the right-hand side is positive
for sufficiently small \(t\). This proves Eq.~\eqref{eq:cuspC}.

We now pass from concurrence to entanglement of formation.
As \(\rho\) is entangled and mixed, \(0<\C(\rho)<1\). Since \(\mathcal E\) is smooth on \((0,1)\),
Eq.~\eqref{eq:cuspC} gives
\[
\begin{aligned}
 \EF(\rho)-\EF(\rho_t)
 &=
 \mathcal E'(\C(\rho))
 \bigl[\C(\rho)-\C(\rho_t)\bigr]+O(t)\\
 &=
 2\mathcal E'(\C(\rho))\kappa_\rho\sqrt t+O(t),
\end{aligned}
\]
which is Eq.~\eqref{eq:cuspE}.
The remainder is \(O(t)\), since \(\C(\rho_t)-\C(\rho)=O(\sqrt t)\).

Combining this expansion with Eq.~\eqref{eq:tracedistance},
we obtain
\[
\begin{aligned}
 \frac{\EF(\rho)-\EF(\rho_t)}{\|\rho_t-\rho\|_1}
 &=
 \frac{\mathcal E'(\C(\rho))\kappa_\rho}{\sqrt t}+O(1)\\
 &\longrightarrow+\infty
 \qquad(t\to0^+),
\end{aligned}
\]
because \(\mathcal E'(\C(\rho))>0\) and \(\kappa_\rho>0\).
Hence Eq.~\eqref{eq:lowerlip} cannot hold with any finite constant.
By the Holevo--Shirokov criterion, no global supporting affine
functional exists at \(\rho\). This completes the converse part of
Theorem~\ref{thm:main}.

\section{Local-unitary invariance and examples}

By Eq.~\eqref{eq:kappa}, \(\kappa_\rho\) depends only on the state and
the radical ray. Let \(L=U_A\otimes U_B\), with
\(U_A,U_B\in U(2)\). The identity
\begin{equation}
 L^\T JL=\chi J,\qquad
 \chi=\det(U_A)\det(U_B),\qquad |\chi|=1
 \label{eq:LUform}
\end{equation}
implies \(\beta(Lx,Ly)=\chi\beta(x,y)\).
Thus, for \(\rho'=L\rho L^\dagger\) with support \(S'=LS\),
\(\Rad(\beta|_{S'})=L\Rad(\beta|_S).\)

Choose \(z'=Lz\). Since
\((\rho'_{S'})^{-1}=L\rho_S^{-1}L^\dagger,\)
we have
\[
 \kappa_{\rho'}
 =\bigl\langle Lz\big|L\rho_S^{-1}L^\dagger\big|Lz\bigr\rangle^{-1/2}
 =\kappa_\rho.
\]

Equation~\eqref{eq:LUform} and \(\Theta=JK\) give \(\Theta L=\overline{\chi}\,L\Theta\).
Hence \(w'=\Theta z'=\overline{\chi}\,Lw\) and therefore \(|w'\rangle\langle w'|=L|w\rangle\langle w|L^\dagger\).
It follows that
\[
 \rho'_t=(1-t)\rho'+t|w'\rangle\langle w'|
 =L\rho_tL^\dagger.
\]
Concurrence is invariant under local unitary transformations. Hence the
coefficient \(2\mathcal E'(\C(\rho))\kappa_\rho\) in Eq.~\eqref{eq:cuspE} is unchanged.

As a rank-two example, consider
\begin{equation}
 \rho_p=p|\Phi^+\rangle\langle\Phi^+|
 +(1-p)|01\rangle\langle01|,
 \qquad 0<p<1,
 \label{eq:example}
\end{equation}
where \(|\Phi^+\rangle=(|00\rangle+|11\rangle)/\sqrt2\).
At \(p=1/2\), this is the counterexample of Holevo and Shirokov
\cite{HolevoShirokov2026}. The radical is spanned by
\(z=|01\rangle\), so \(w=\Theta z=|10\rangle\) and \(\kappa_{\rho_p}=\sqrt{1-p}\). The concurrence is \(\C(\rho_p)=p\).

Adding a population \(t\) in \(|10\rangle\), let \(\rho_{p,t}=(1-t)\rho_p+t|10\rangle\langle10|\). For sufficiently small \(t\),
\begin{equation}
 \C(\rho_{p,t})
 =(1-t)p-2\sqrt{t(1-t)(1-p)}.
 \label{eq:exampleexact}
\end{equation}
Thus the coefficient of the \(\sqrt{t}\) term is
\(2\sqrt{1-p}\), in agreement with Eq.~\eqref{eq:cuspC}.

For rank three, consider an entangled state supported on \(S=\operatorname{span}\{|00\rangle,|01\rangle,|11\rangle\}\).
Its kernel is spanned by \(|10\rangle\), while the radical of
\(\beta|_S\) is spanned by \(|01\rangle\). Hence we take \(z=|01\rangle\) and \(w=|10\rangle\). Ordering the basis as \(\{|00\rangle,|11\rangle,|01\rangle\}\), write
\[
 \rho_S=
 \begin{pmatrix}
 G&v\\
 v^\dagger&d
 \end{pmatrix}>0,
\]
where \(G\) is a positive-definite \(2\times2\) matrix,
\(d\) is the population in \(|01\rangle\), and \(v\) contains
the coherences between \(|01\rangle\) and
\(|00\rangle,|11\rangle\). By block inversion,
\[
 \kappa_\rho^2
 =\bigl(\langle01|\rho_S^{-1}|01\rangle\bigr)^{-1}
 =d-v^\dagger G^{-1}v>0.
\]
Thus \(\kappa_\rho\) depends not only on the population \(d\) but also
on the coherences contained in \(v\). In particular, states with the
same support need not have the same value of \(\kappa_\rho\). The
leading coefficient for the decrease in the entanglement of formation is \(2\mathcal E'(\C(\rho))\kappa_\rho\).

\section{Conclusion}

We have shown that, for an entangled two-qubit state, a global
supporting affine functional for the entanglement of formation
exists exactly when the spin-flip bilinear form is nondegenerate
on its support. The same condition is expressed by the nonvanishing
of the last concurrence root associated with the support.
In the nondegenerate case, the entanglement of formation satisfies
a two-sided Lipschitz bound relative to the reference state.
If the restricted form is degenerate, suitable perturbations
outside the support give a square-root decrease of the entanglement
of formation. Thus the one-sided Lipschitz bound does not hold.

The degenerate supports have a simple geometric description. The projective
support is a tangent line to the product-state quadric at rank two and a
tangent plane at rank three. In the latter case, this is equivalent to the
kernel being spanned by a product vector. Entangled states
of rank one and rank four are nondegenerate. The support determines
whether the square-root behavior occurs. The coefficient, however,
depends on the state and is unchanged under local unitary
transformations.

The results show, in particular, that properties of the convex roof
restricted to a fixed support do not determine the behavior of the
entanglement of formation under perturbations outside this support.
These perturbations have to be taken into account when one asks
whether a local supporting functional can be extended to a global
supporting functional.

\end{document}